\documentclass[11pt]{article}
\usepackage{latexsym}
\usepackage{mathrsfs}
\usepackage{amsmath,amsthm}
\usepackage{graphicx,epstopdf,epsfig,multirow,epic,bm}
\usepackage{color,fancyvrb}
\usepackage{colortbl}
\usepackage{multicol}
\usepackage{amssymb}

\usepackage{esint}  

\usepackage{makeidx}
\usepackage{showidx}
\usepackage{textcomp}
\usepackage{paralist}

\usepackage{CJK}

\theoremstyle{plain}
\newtheorem{thm}{Theorem}

\theoremstyle{definition}

\theoremstyle{plain}

\theoremstyle{problem}

\theoremstyle{plain}

\theoremstyle{plain}

\theoremstyle{plain}

\DeclareGraphicsExtensions{.eps,.eps.gz}
\normalsize \rm
\begin{document}

\thispagestyle{empty}

\begin{center}
{\Large\bf Maximum Leaf Spanning Trees of Growing Sierpinski Networks Models}\\[12pt]
{\large Bing \textsc{Yao}$^{1}$, Xia \textsc{Liu}$^{1}$, Jin \textsc{Xu}$^{2}$}
\end{center}
\begin{flushleft}
{\footnotesize 1. College of Mathematics and Statistics, Northwest
Normal University, Lanzhou, 730070, China\\
2. School of Electronics Engineering and Computer Science, Peking University, Beijing, 100871, China}
\end{flushleft}

\begin{abstract}
The dynamical phenomena of complex networks are very difficult to predict from local information due to the rich microstructures and corresponding complex dynamics. On the other hands, it is a horrible job to compute some stochastic parameters of a large network having thousand and thousand nodes. We design several recursive algorithms for finding spanning trees having maximal leaves (MLS-trees) in investigation of topological structures of Sierpinski growing network models, and use MLS-trees to determine the kernels, dominating and balanced sets of the models. We propose a new stochastic method for the models, called the edge-cumulative distribution, and show that it obeys a power law distribution.
\end{abstract}

\section{Introduction}

In understanding complex networks, one must know the global properties of networks as well as the local properties such as the degree distribution. Bollob\'{a}s and Riordan Ref.\cite{Bollobas-Riordan2004} consider a random graph process in which vertices are added to the graph one at a time and joined to a fixed number $m$ of earlier vertices, where each earlier vertex is chosen with probability proportional to its degree. This process was introduced by Barab\'{a}si and Albert in Ref. \cite{A-L-Barabasi-R-Albert1999}, as a simple model of the growth of real-world graphs such as the world-wide web.

It is well known that the dynamical phenomena of networks such as traffic and information flow are very difficult to predict from local information due to the rich microstructures and corresponding complex dynamics (Ref. \cite{Kim-Noh-Jeong2004}). On the other hands, it is difficult to implement miscellaneous measurements on complex networks, such as the betweenness centrality (BC) $b(\bullet)=\sum_{i\neq j} \frac{b(i,\bullet,j )}{b(i,j)}$, where $\bullet$ is a node or an edge, and $b(i,\bullet,j )$ is the number of shortest $(i,j)$-paths through $\bullet$ and $b(i,j)$ is the number of shortest $(i,j)$-paths. Clearly, this is a horrible job to compute the betweenness centrality of a network having thousand and thousand nodes. But Kim \emph{et al.} (Ref. \cite{Kim-Noh-Jeong2004}), after studying the properties of the spanning trees with maximum total edge betweenness centrality, point out that the scale-free spanning trees represent the \emph{communication kernels} on networks, and the scale-free spanning trees show robust characteristics in the degree correlation and the betweenness centrality distribution. They defined the communication kernel of a network as the spanning tree with a set of edges maximizing the summation of their edge BC's on the original networks, and investigated the structural and dynamical properties of the spanning tree of complex networks and the role of shortcuts in the networks, and find that the spanning trees show scale-free behavior in the degree distributions.

As known, the Maximum Leaf Spanning Tree (MLS-tree) problem, which asks to find, for a given graph, a spanning tree with as many leaves as possible, is one of the classical NP-complete problems in Ref. \cite{Garey-Johnson1979}. Fernau \emph{et al.} (Ref. \cite{Fernau-Kneis-Kratsch-Langer-Liedloff-Raible-Rossmanith2011}) investigated MLS-trees based on an exponential time viewpoint that is equivalent to the Connected Dominating Set problem (CDSP), and present a branching algorithm whose running time of $O(1.8966^n)$ has been analyzed using the Measure-and-Conquer technique as well as a lower bound of
$\Omega(1.4422^n)$ for the worst case running time of their algorithm. By means of MLS-trees of growing networks in Ref. \cite{Yao-Liu-Zhang-Chen-Zhang-Yao-Zhao2013}, a \emph{stochastic network model} can be defined as
$M_t=(p(u, k,t)$, $G(t)$, $UG)$ for $t\in [a, b]$, where $UG$ is the \emph{underlying graph} of $M_t$ that contains
all nodes and links appeared in $M_t$ for $t\in [a,b]$; $p(u, k,t)$ is the \emph{probability} of a node $u$ being
connected with other $k$ nodes in $M_{t'}~(t\,'\in [a,t))$; $G(t)$ is the \emph{connected topological graph} of $
M_t$. A node of $M_t$ is called an \emph{alltime-hub node} if it is not
a leaf of any MLS-tree of $M_t$ at $t\in [a,b]$. A \emph{kernel} of $M_t$ is an induced graph over the set of alltime-hub nodes.

We show our algorithms to find  MLS-trees for investigating topological structures of growing Sierpinski network models (GSN-models) introduced in Ref. \cite{Zhang-Zhou-Fang-Guan-Zhang2007}. All graphs mentioned here are simple, undirected and finite. A \emph{leaf} is a node of degree one. For a graph $G$, we let $L(G)$ stand for the set of its leaves, and $D(G)$ be the diameter of $G$, that is, $D(G)=\min_{i\neq j}\{d(i,j)\}$, where $d(i,j)$ is the \emph{geodesic distance} from node $i$ to node $j$. Notation $|X|$ is the number of elements of a set $X$.

\section{GSN-models}

The graph $O$ in Figure \ref{fig.1} shows the result of a \emph{fractal-operation} that will be used in the following. For a given triangle $\Delta ABC$ with three nodes $A,B,C$ in the plane, we embed another triangle $\Delta abc$ with three nodes $a,b,c$ in the inner face of $\Delta ABC$, and then join $A$ with $b$ and $c$ to produce two edges $Ab,Ac$; join $B$ with $c$ and $a$ to form two edges $Bc,Ba$; and join $C$ with $a$ and $b$ to generate two edges $Ca,Cb$. The resulting graph $O$ is called the \emph{base}, and it has six inner triangles that bound six regions $I$, $III'$ $II$, $I'$, $III$, $II'$ in the clockwise direction. For simpler statement, we call $A,B,C$ three \emph{major nodes} of the base $O$, where $A$ is called the \emph{left major node}, $B$ the \emph{top major node}, and $C$ the \emph{right major node}, and we call $a,b,c$ to be three \emph{submajor nodes}. We restate the construction of the networks $N(t)$ shown in Ref. \cite{Zhang-Zhou-Fang-Guan-Zhang2007} by adding a labelling function below.

\begin{figure}[h]
\centering
\includegraphics[height=3.8cm]{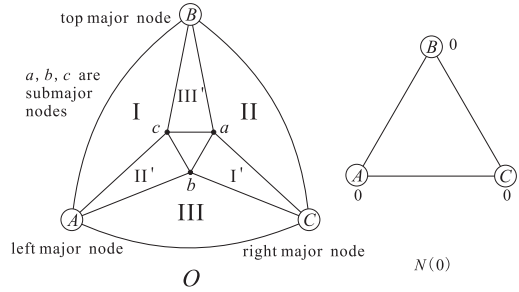}
\caption{\label{fig.1} The left is for illustrating a fractal-operation, and the right is the initial network.}
\end{figure}

\subsection{Construction of GSN-models} Let $N(0)$ be the initial network pictured in Figure \ref{fig.1}, and let $V(0)$ be the node set of $N(0)$. We define a labelling $f$ such that $f(\alpha)=0$ for each node $\alpha\in V(0)$. Do a fractal-operation to the inner face of $N(0)$ by adding a new triangle produces the second GSN-model $N(1)$, and label $f(\beta)=1$ for every node $\beta\in V(1)\setminus V(0)$. To form the third  GSN-model $N(2)$ from $N(1)$, we do a fractal-operation to each inner triangle $\Delta uvw$ of $N(1)$ without $f(u)=f(v)=f(w)$, and label each node $x\in V(2)\setminus V(1)$ as $f(x)=2$. Go on in this way, every GSN-model $N(t)$ can be obtained from the previous GSN-model $N(t-1)$ for $t\geq 2$ by doing a fractal-operation to each inner triangle $\Delta xyz$ of $N(t-1)$ without $f(x)=f(y)=f(z)$, and label each node $w\in V(t)\setminus V(t-1)$ as $f(w)=t$.

Let $n_v(t)$ and $n_e(t)$ be the numbers of nodes and edges of the network $N(t)$, respectively. Clearly, each $N(t)$ has an outer face $\Delta ABC$ and an inner face $\Delta abc$ for $t\geq 1$ (see Figure \ref{fig.2}). Another way to generate $N(t)$ from  $N(t-1)$ is the generalized MLS-TREE algorithm introduced later.

\begin{figure}[h]
\centering
\includegraphics[height=6cm]{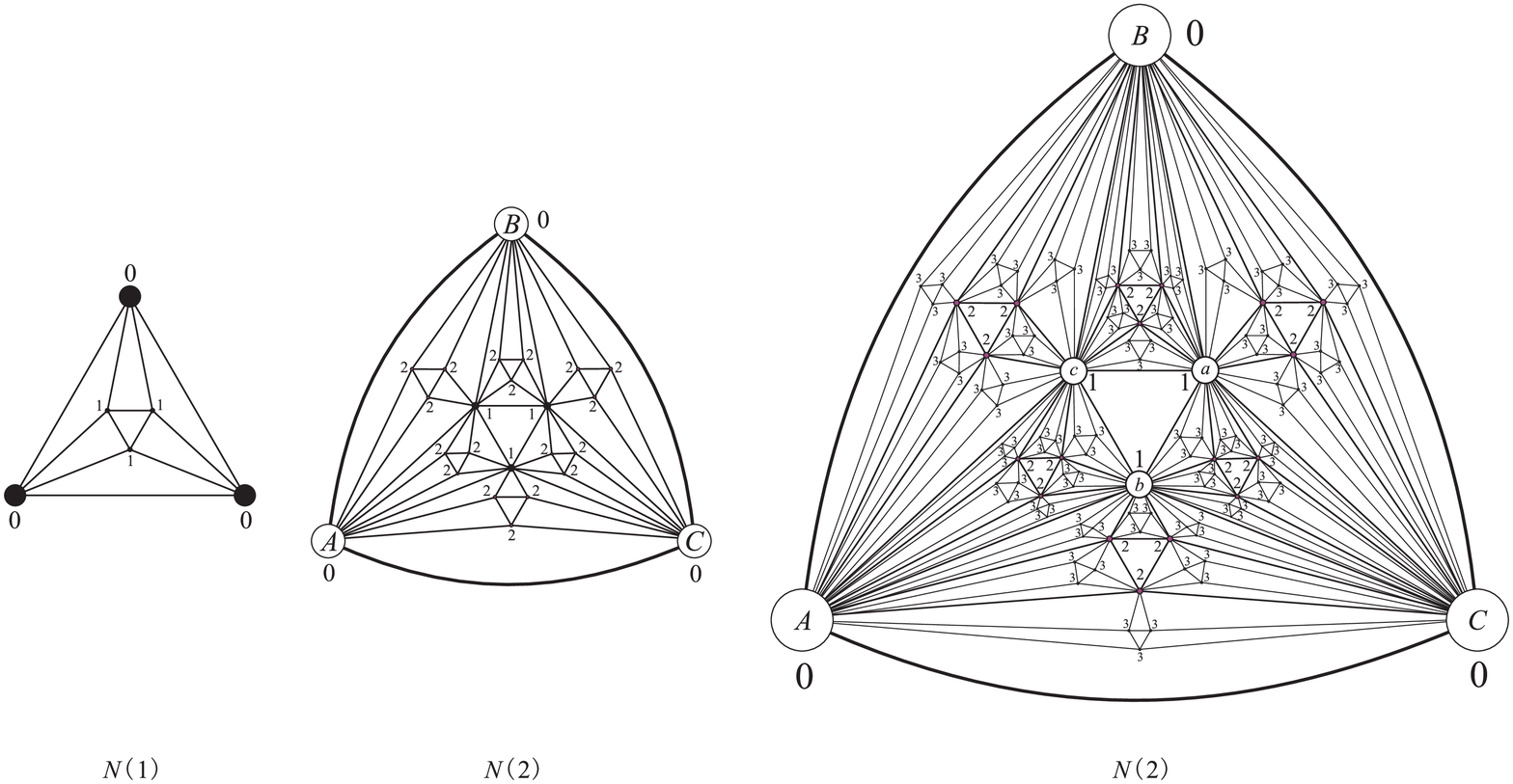}
\caption{Three  GSN-models  $N(1)$, $N(2)$ and $N(3)$.}
\label{fig.2}
\end{figure}

\section{Deterministic statistics of GSN-models}

Four first GSN-models $N(0)$, $N(1)$, $N(2)$ and $N(3)$ are shown in Figures \ref{fig.1} and \ref{fig.2}, respectively.  Notation $n_d(t)$ stands for the number of nodes of degree $d$ and let $\Delta(t)$ be the maximum degree in $N(t)$. For $t\geq 2$, the \emph{degree spectrum} of $N(t)$ is that each number $n_d(t)$ of nodes of degree $d=1+3^{k}$ is equal to $3\cdot 6^{t-k}$ for $k=1,2,\dots ,t-1$, respectively; and the number of nodes of maximum degree $\Delta(t)=1+3^{t}$ is $n_{\Delta(t)}(t)=6$ that are the major nodes $A,B,C$ and the subnodes $a,b,c$. It is easy to show $n_v(t)$ in the formula (\ref{eqa:nodes-edges}) by the degree spectrum of $N(t)$. Since $N(t)$ is a maximal planar graph, so $n_e(t)=3n_v(t)-6$,
\begin{equation}\label{eqa:nodes-edges}
n_v(t)=(3\cdot 6^t+12)/5,\quad n_e(t)=(9\cdot 6^t+6)/5
\end{equation}
as well as there are $n_e(t)-n_v(t)+1$ inner triangles in $N(t)$ by the famous Euler's formula on planar graphs. By the degree spectrum of $N(t)$ and by the \emph{linear preferential attachment rule} (Ref.\cite{A-L-Barabasi-R-Albert1999}, \cite{Costa-Rodrigues-Travieso-Villas-Boas2006}), we can get the probability of joining a new node $u$ out of $N(t)$ to a node $v_{d(t)}$ of degree $d(t)$ in $N(t)$ as follows, $d(t)=k$,
\begin{equation}\label{eqa:linear-preferential-attachment-rule}
 P(u\rightarrow v_{d(t)})=\frac{k(v_{d(t)})}{\sum_{w\in V(t)} k(w)}=\frac{3^{k}+1}{2n_e(t)}\approx\frac{5}{18}\cdot \frac{3^{k}}{6^t},
\end{equation}
where $k(x)$ is the degree of node $x$ in $N(t)$. Again we obtain $3^{t-k}\cdot P(u\rightarrow v_{d(t)})=P(u\rightarrow v_{\Delta(t)})$ for large integers $t>0$, which means that the nodes of large degrees play the role like hubs connecting the whole network together. This phenomenon is known as ``the rich get richer'' paradigm.

The \emph{average (mean) degree} $\langle k\rangle$ of $N(t)$ is defined as
\begin{equation}\label{eqa:average-degree}
\langle k\rangle=\frac{2n_e(t)}{n_v(t)}=\frac{1}{n_v(t)}\sum _{v\in V(t)}k(v).
\end{equation}
The \emph{average-square (mean-square) degree} $\langle k^2\rangle$ of $N(t)$ is determined by $\langle k^2\rangle=\frac{1}{n_v(t)}\sum _{v\in V(t)}k^2(v)$. Newman \cite{M-E-J-Newman2008} pointed that a giant component exists in the network if and only if $\langle k^2\rangle -2\langle k\rangle > 0$. We verify that $N(t)$ holds $\langle k^2\rangle -2\langle k\rangle > 0$ when $t\geq 2$.

\subsection{The edge-cumulative distribution} Motivated from the cumulative degree distribution that is an important character of scale-free networks (Refs. \cite{M-E-J-Newman2003}, \cite{Zhang-Zhou-Fang-Guan-Zhang2007}), we propose a deterministic statistic for $2<t_i<t$, named as the \emph{edge-cumulative distribution} $P_{\textrm{e-cum}}(k)$, as follows
\begin{equation}\label{eqa:4444}
{
\begin{split}
&\quad P_{\textrm{e-cum}}(k)=\frac{1}{n_e(t)}\sum^{t_i}_{j=0} n_e(j)\\
&=\frac{1}{9\cdot 6^t+6}\sum^{t_i}_{j=0} (9\cdot 6^j+6)\\
&=\frac{1}{9\cdot 6^t+6}\left [15+6t_i+\frac{54}{5}(6^{t_i}-1)\right ]\approx \frac{6}{5}6^{t_i-t}
\end{split}}
\end{equation}
Plugging $t_i=t-\frac{\ln k}{\ln 3}$ into Eq. (\ref{eqa:4444}) leads to $P_{\textrm{e-cum}}(k)\propto \frac{6}{5} k^{-1-\ln 2/\ln 3}$, which means that $P_{\textrm{e-cum}}(k)$ follows a power law form with the exponent $\gamma_k=1+\frac{\ln 2}{\ln 3}$.

\subsection{$(\alpha_k,\beta_k)$-GSN models} By the degree spectrum of a GSN-model $N(t)$, for $4\leq d\leq 3^{k}+1$,  adding numbers of nodes of degree $d\leq k$ together is $S(\leq k)=\sum_{d\leq k}n_d(t)=\sum^k_{i=1}3\cdot 6^{t-i}=\frac{3}{5}6^{t-k}(6^{k}-1)$, and adding  degrees of nodes of degree $d\leq k$ is equal to $D(\leq k)=\sum_{d\leq k}d\cdot n_d(t)=\sum^k_{i=1}3\cdot 6^{t-i}(3^{i}+1)=\frac{3}{5}6^{t}(6-\frac{1}{6^{k}}-\frac{5}{2^{k}})$. Thereby, the sum $S(\geq k+1)$ of numbers of nodes of degree $d$ with $3^{k+1}+1\leq d\leq 3^{t}+1$ is $S(\geq k+1)=n_v(t)-S(\leq k)$, and the sum $D(\geq k+1)$ of their degrees is equal to $D(\geq k+1)=2n_e(t)-D(\leq k)$. If the \emph{node-number proportion} $\frac{S(\geq k+1)}{n_v(t)}=\alpha_k$, so we have
\begin{equation}\label{eqa:c3xxxxx}
\alpha_k=\frac{3\cdot 6^{t-k}+12}{3\cdot 6^{t}+12}\ \sim \ 6^k=\frac{1}{\alpha_k}.
\end{equation}
we get $k=-\frac{\ln \alpha_k}{\ln 6}$. From the node-number proportion $\frac{S(\leq k)}{n_v(t)}=1-\alpha_k$, we solve $ k=-\frac{\ln \alpha_k}{\ln 6}$ too. We call $N(t)$ an \emph{$(\alpha_k,\beta_k)$-GSN-model}, where $\beta_k=\frac{D(\geq k+1)}{2n_e(t)}=1-\frac{D(\leq k)}{2n_e(t)}$, and furthermore $S(\geq k+1)D(\geq k+1)=2\alpha_k \beta_k n_v(t)n_e(t)$. As a test, we take $\alpha_k=\frac{1}{2}\cdot \frac{1}{10^6}$, so $k\approx 8.0974$ and the \emph{node-degree proportions} are
\begin{equation}\label{eqa:new-statistics}
\frac{D(\leq k)}{2n_e(t)} \approx 1-\frac{1}{6}\left (\frac{1}{6^{k}}+\frac{5}{2^{k}}\right )\approx 0.997,
\end{equation}
and $\beta_k \approx 0.003$. The parameters $\alpha_k,\beta_k$ show a description of a GSN-model as: The nodes having degrees$\leq k$ show a powerful controlling almost edges of $N(t)$ as $k$ is smaller. Conversely, the nodes with degrees$\geq k+1$ are incident to fewer edges of $N(t)$, however, they connect the the network model together. For example, deleting the nodes of $N(k)$ from $N(t)$ with $t>k\geq 1$ produces $6^{k-1}$ fragments. Moreover, the deletion of nodes of $V(t-1)$ from $N(t)$ make the remainder consisting of $6^{t-1}$ components (triangles) in which the total number of nodes is equal to $3\cdot 6^{t-1}=5n_v(t-1)-12$. Obviously, such a deletion destroys the network made by a GSN-model into pieces, since $3\cdot 6^{t-1}/n_v(t)=\frac{5}{6}$.

\section{MLS-trees of GSN-models}

In this section we will determine the \emph{kernels} of GSN-models by MLS-trees in  the models, and furthermore we present several algorithms for finding MLS-trees. The notation $T^M(t)$ denotes an MLS-tree of $N(t)$, and $L(T^M(t))$ stands for the set of leaves of $T^M(t)$. It is easy to verify that $|L(T^M(0))|=2$, $|L(T^M(1))|=4$ and $|L(T^M(2))|=19$.

\subsection{Construction of MLS-trees of GSN-models}

Let $M_{st}(t)$ be the set of MLS-trees of $N(t)$ such that three major nodes $A,B,C$ of every MLS-tree of $M_{st}(t)$ are not leaves, and every one of $M_{st}(t)$ has two edges $AB,BC$ and has no the edge $AC$. Since an MLS-tree of $N(t)$ can be constructed by different MLS-trees of $N(t-1)$, so we call them \emph{non-uniformly MLS-trees}.

\begin{flushleft}
\textbf{Generalized MLS-TREE algorithm}
\end{flushleft}

\textbf{Input:} A GSN-model $N(t)$ for $t\geq 3$, $M_{st}(2)=\{T^M_{i}(2):i=1,2,\dots, m_2\}$, where each $i$ is called a \emph{footscript}, and $m_2$ is the number of elements of $M_{st}(2)$.

\textbf{Output:} $M_{st}(t)$.

1. $M_{st}(2)\leftarrow M_{st}(2)$, $P_2\leftarrow \{(k_1,k_2,\dots ,k_6)\}$, where $(k_1,k_2,\dots ,k_6)$ are permutations over all footscripts of MLS-trees of $M_{st}(2)$; $i\leftarrow 2$.

2. If $i<t$, $M_{st}(i+1)\leftarrow \emptyset$, go to 3; otherwise $M_{st}(t)\leftarrow M_{st}(i)$, go to 5.

3. If $P_i\neq \emptyset $, go to 4; otherwise $P_{i+1}\leftarrow \{(k_1$, $k_2$, $\dots $, $k_6)\}$, where $(k_1,k_2,\dots ,k_6)$ is a permutation of footscripts of MLS-trees of $M_{st}(i)$; $i\leftarrow i+1$ go to 4.

4. Take a permutation $Q\in P_i$, \textbf{do} (Ref. the graph $O$ of Figure \ref{fig.1} and six graphs shown in  Figure \ref{fig.3}):

(a1) rename $T^M_{k_1}(i)\in M_{st}(i)$ as $I$-$T^M_{k_1}(i)$ and set its left major node $c\leftarrow A$, its top major node $A\leftarrow B$, and its right major node $B\leftarrow C$;

(a2) delete two edges $AB,BC$ of $T^M_{k_2}(i)\in M_{st}(i)$, and name the remainder as $III'$-$T^M_{k_2}(i)$, and then set its left major node $c\leftarrow A$, its top major node $B\leftarrow B$, and its right major node $a\leftarrow C$;

(a3) rename $T^M_{k_3}(i)\in M_{st}(i)$ as $II$-$T^M_{k_3}(i)$, and set its left major node $a\leftarrow A$, its top major node $B\leftarrow B$, and its right major node $C\leftarrow C$;

(a4) delete two edges $AB,BC$ of $T^M_{k_4}(i)\in M_{st}(i)$ and, name the remainder as $I'$-$T^M_{k_4}(i)$, and set its left major node $a\leftarrow A$, its top major node $C\leftarrow B$, and its right major node $b\leftarrow C$;

(a5) rename $T^M_{k_5}(i)\in M_{st}(i)$ as $III$-$T^M_{k_5}(i)$, and set its left major node $b\leftarrow A$, its top major node $C\leftarrow B$, and its right major node $A\leftarrow C$;

(a6) delete its two edges $AB,BC$ of $T^M_{k_6}(i)\in M_{st}(i)$ and, name the remainder as $II'$-$T^M_{k_6}(i)$, and set its left major node $b\leftarrow A$, its top major node $A\leftarrow B$, and its right major node $c\leftarrow C$.

Identify the major nodes of the above six graphs $I$-$T^M_{k_1}(i)$, $III'$-$T^M_{k_2}(i)$, $II$-$T^M_{k_3}(i)$, $I'$-$T^M_{k_4}(i)$, $III$-$T^M_{k_5}(i)$ and $II'$-$T^M_{k_6}(i)$ having the same letters into one node, respectively. The resulting is just an MLS-tree $T^M(i+1)$ of $N(i+1)$, and $T^M(i+1)$ has three submajor nodes $a,b,c$ and three major nodes $A,B,C$ (Ref. Figure \ref{fig.MLS-TREE-algorithm}).

Let $f(A)\leftarrow 0$, $f(B)\leftarrow 0$, $f(C)\leftarrow 0$, $f(a)\leftarrow 1$, $f(b)\leftarrow 1$, $f(c)\leftarrow 1$; $f(x)\leftarrow f(x)+1$ for $x\in V(T^M(i+1))\setminus \{a,b,c,A,B,C\}$; $M_{st}(i+1)\leftarrow M_{st}(i+1)\cup \{T^M(i+1)\}$, $P_i\leftarrow P_i\setminus \{Q\}$, go to 3.

5. return $M_{st}(t)$.

\begin{figure}[h]
\centering
\includegraphics[height=5cm]{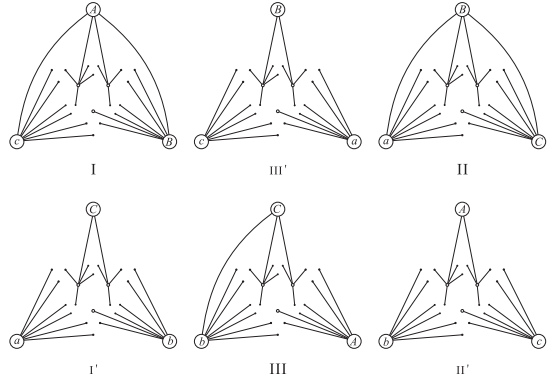}
\caption{Six graphs for illustrating the generalized MLS-TREE algorithm.}
\label{fig.3}
\end{figure}

\begin{figure}[h]
\centering
\includegraphics[height=5cm]{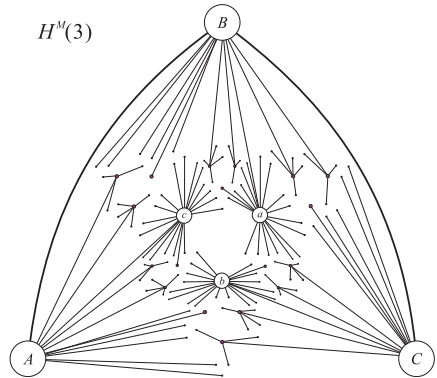}
\caption{An MLS-tree $H^M(3)$ obtained by six graphs shown in Figure \ref{fig.3} and the  generalized MLS-TREE algorithm.}
\label{fig.MLS-TREE-algorithm}
\end{figure}

\begin{thm} \label{them:MLS-tree-leaves}
Every MLS-tree of a GSN-model $N(t)$ for $t\geq 3$ has $19\cdot 6^{t-2}$ leaves.
\end{thm}
\begin{proof}  It is not hard to verify that each MLS-tree of $N(t)$ has  $19\cdot 6^{k-2}$ leaves for $k=2,3$. So, we can use the induction on time steps $t\geq 3$. In the generalized MLS-TREE algorithm, we can select an MLS-tree $T^M(3)$ having $19\cdot 6^{3-2}$ leaves and take $k_i=k_j$ for all permutations $(k_1,k_2,\dots ,k_6)$ at each time step to obtain an MLS-tree $T^M(t)$ having $19\cdot 6^{t-2}$ leaves for $t\geq 4$.
Note that $N(t)$ is the resulting of overlapping six models $N_i(t-1)$ ($\cong N(t-1)$) for $i=1,2,\dots,6$, according to the generalized MLS-TREE algorithm.
If $N(t)$ has a spanning tree $G$ such that $|L(G)|>19\cdot 6^{t-2}$, so $G$ induces six spanning trees $G_i$ of $N_i(t-1)$ for $i=1,2,\dots,6$ by the generalized MLS-TREE algorithm. We have one spanning tree $G_i$ with $|L(G_i)|>19\cdot 6^{t-2}$, which contradicts with the induction hypothesis.
\end{proof}

The number $|L(T^M(t))|$ of leaves of a MLS-tree $T^M(t)$ of $N(t)$ is $19\cdot 6^{t-2}=\frac{95}{18}n_v(t-1)-\frac{38}{3}$, and also greater than the number $n_v(t)-n_v(t-1)$ of adding new nodes to $N(t-1)$ since $19\cdot 6^{t-2}=6^{t-2}+n_v(t)-n_v(t-1)$.

\subsection{Balanced sets in GSN-models}

In Ref. \cite{Yao-Zhang-Wang2010}, a non-empty node subset $X$ of node set $V(t)$ of a GSN-model $N(t)$ is called a \emph{$T^M$-balanced set} if
\begin{equation}\label{eqa:maxi-balanced-set-definition}
|X\cap T^M_i(t)|=|X\cap T^M_j(t)|
\end{equation}
holds for any pair of MLS-trees $T^M_i(t)$ and $T^M_j(t)$ of $N(t)$. Obviously, $V(t)$ holds Eq.(\ref{eqa:maxi-balanced-set-definition}) true, but it is trivial. Our goal is to find $T^M$-balanced sets $X\neq V(t)$, and such sets are called \emph{proper $T^M$-balanced sets}.

\begin{thm} \label{them:basic-maximum-leaf-trees}
For every GSN-model $N(t)$ with $t\geq 3+k$ and $k\geq 0$, each node of $N(k)$ is not a leaf of any MLS-tree of $N(t)$.
\end{thm}
\begin{proof} It is not hard to see that $V(0)$ is not a proper $T^M$-balanced set of $N(1)$ and $N(2)$ (Ref. these two models $N(1)$ and $N(2)$ shown in Figure \ref{fig.2}).

\emph{Case 1.} $k=0$ and $t\geq 3$. For any MLS-tree $T^M(t)$ of $N(t)$, we say $V(0)\cap L(T^M(t))=\emptyset$ when $t\geq 3$. If it is not so, without loss of generality, the left major node $A\in V(0)\cap L(T^M(t))$, namely, the node $A$ is a leaf of a MLS-tree $T^M(t)$ of the GSN-model $N(t)$. In the region $II'$, there are triangles $\Delta x^{(r)}_iy^{(r)}_iz^{(r)}_i$ of $N(r)$ holding distances $d(A,x^{(r)}_i)=2$ and $d(A,y^{(r)}_i)=d(A,z^{(r)}_i)=1$ for $i=1,2,\dots, 3^{r-1}$ with $2\leq r\leq t$. In $N(t)$ and $N(t-1)$, there are two triangles $\Delta x^{(t)}_iy^{(t)}_iz^{(t)}_i$ and $\Delta x^{(t-1)}_jy^{(t-1)}_jz^{(t-1)}_j$ such that three nodes of $\Delta x^{(t)}_iy^{(t)}_iz^{(t)}_i$ are only connected with nodes $y^{(t-1)}_j$, $z^{(t-1)}_j$ and $A$. In other words, nodes $y^{(t-1)}_j$, $z^{(t-1)}_j$ are not leaves of $T^M(t)$. We can make another spanning tree $H$ of $N(t)$ by joining $A$ with $y^{(t)}_i,z^{(t)}_i$, joining $y^{(t-1)}_j$ with $x^{(t)}_i$, and make $z^{(t-1)}_j$ to be a leaf of the new spanning tree. Note that we can do the above work in the rest regions $I$ and $III$. Eventually, we obtain a spanning tree $H^*$ of $N(t)$ such that $|L(H^*)|\geq 2+L(T^M(t))$, which violates the definition of $T^M(t)$ having maximal leaves in $N(t)$.

\emph{Case 2.} $k\geq 1$. The constraint $t\geq 3+k$ is necessary, because there exists an MLS-tree $T^M(3)$ of the GSN-model $N(3)$ such that $(V(1)\setminus V(0))\cap L(T^M(3))\neq \emptyset$. Furthermore, based on $T^M_1(3)$, we can use the generalized MLS-TREE algorithm to obtain an MLS-trees $H(m)$ such that $(V(m-2)\setminus V(m-3))\cap L(H(m))\neq \emptyset$ with $m\geq 3$.

By contradiction to show that each node of $V(k)$ is not a leaf of any MLS-tree $T^M(t)$ with $t\geq 3+k$. Assume that $N(t)$ has an MLS-tree $T^M(t)$ such that $V(k)\cap L(T^M(t))\neq \emptyset$ for some $k$ with $1\leq k\leq t-3$. So, we have a leaf $x^{(k)}_s\in V(k)\cap L(T^M(t))$ and a triangle $\Delta x^{(k)}_sy^{(k)}_sz^{(k)}_s$. By the structure of $N(t)$, there is a triangle $\Delta x^{(k+1)}_sy^{(k+1)}_sz^{(k+1)}_s$ inside of the triangle $\Delta x^{(k)}_sy^{(k)}_sz^{(k)}_s$ such that these two triangles form a subgraph $O^*$ like the graph $O$ shown in Figure \ref{fig.1}. In this case, the position of the node $x^{(k)}_i$ is as the same as the node $A$ in Case 1. In the region $II'$ of $O^*$, there are triangles $\Delta x^{(r)}_iy^{(r)}_iz^{(r)}_i$ of $N(r)$ holding distances $d(x^{(k)}_s,x^{(r)}_i)=2$ and $d(x^{(k)}_s,y^{(r)}_i)=d(x^{(k)}_s,z^{(r)}_i)=1$ for $i=1,2,\dots, 3^{r-k-1}$ with $k+2\leq r\leq t$. Based on the same proof shown in Case 1, we can get a contradiction with the definition of $T^M(t)$.
\end{proof}

A subset $S$ of node set $V(t)$ of a GSN-model $N(t)$ is called a \emph{dominating set} if every node of $V(t)$ is adjacent to a node of $S$ or belongs to $S$. By Theorem \ref{them:basic-maximum-leaf-trees} we can confirm the following results

\begin{thm} \label{them:kernel-balanced}
Every GSN-model $N(t)$ with $t\geq k+3$ and $k\geq 0$ holds:

$(i)$ every $V(k)$ is a proper connected $T^M$-balanced set and induces a connected kernel of $N(t)$;

$(ii)$ $V(k)\subset V(T^M_i(t))\cap V(T^M_j(t))$ for any two MLS-trees $T^M_i(t)$ and $T^M_j(t)$ of $N(t)$; and

$(iii)$ $V(t)\setminus L(T^M(t))$ is a connected dominating set of $N(t)$.
\end{thm}

Although the nodes with degrees$\geq k+1$ do not control more edges by Eq. (\ref{eqa:new-statistics}), but they are a controlling center in $N(t)$ with $t\geq 3+k$. By Theorem  \ref{them:kernel-balanced} we can find a maximal proper connected $T^M$-balanced set $X_1$ of $N(t)$, and then get a connected model $I_1$ induced over $X_1$. Next, $I_1$ has a  maximal proper connected $T^M$-balanced set $X_2$ that induces a connected model $I_2$. In this way, we obtain a sequence of models $I_1,I_2,\dots I_m$ such that $I_{i+1}\subset I_i$, and $X_{i+1}$ is a maximal kernel of $I_i$ for $i=1,2,\dots, m-1$. Clearly, $I_i\subset N(t-2i)$.

\subsection{A dynamic algorithm for finding MLS-trees of GSN-models}

There are some algorithms to find spanning trees of networks in Ref. \cite{Yao-Liu-Zhang-Chen-Zhang-Yao-Zhao2013}. We will apply the Bread-first Search algorithm (BFSA) introduced in Ref. \cite{Bondy-Murty-new} to make our Dynamic First-first BFSA algorithm (DFF-BFSA algorithm) by the motivation of the linear preferential attachment rule (Ref.\cite{Costa-Rodrigues-Travieso-Villas-Boas2006}). Predecessors' children are searched before successors' children, according to a rule of ``priority has priority''.

\begin{flushleft}
\textbf{DFF-BFSA algorithm}
\end{flushleft}

\textbf{Input:} A GSN-model $N(t)$ for $t\geq 0$.

\textbf{Output:} A spanning tree $T(t)$.

1. For the GSN-model $N(0)$, BFSA outputs a spanning tree $T(0)$ with $V(0)=\bigcup ^{m(0)}_{j=0}V^{(0)}_j$ and a level function $l$ such that $l(x)=j$ for $x\in V^{(0)}_j$, and the nodes of $V^{(0)}_j$ are ordered well by BFSA for $j=1,2,\dots ,m(0)$.

2. Let $\textrm{nei}(x,k)$ be the neighborhood of a node $x$ of $N(k)$ at time step $k$. At time step $k+1$, $V(k)=\bigcup^{k}_{l=0} \bigcup^{m(l)}_{j=m(l-1)+1} V^{(l)}_j$ (here, $m(-1)=-1$). Implementing BFSA \textbf{do}: For every ordered set $V^{(l)}_j=\{x^{(l)}_{j,1},x^{(l)}_{j,2},\dots ,x^{(l)}_{j,m(l,j)}\}$ with $l\leq k$, from $i=1$ to $i=m(l,j)$, scan $y\in \textrm{nei}(x^{(l)}_{j,i},k+1)\setminus \{w\in V(t): l(w)\textrm{ exists}\}$, $b\leftarrow l(x^{(l)}_{j,i})+1$, $l(y)\leftarrow m(k)+b$, and add $y$ to the ordered set $V^{(k+1)}_{m(k)+b}$ as the last node, add node $y$ and edge $yx^{(l)}_{j,i}$ to the spanning tree $T(k)$ in order to form new spanning tree $T(k+1)$; $m(k+1)\leftarrow \max \{m(k)+l(x)+1:x\in V(k)\}$, $V(k+1)\setminus V(k)\leftarrow \bigcup^{m(k+1)}_{j=m(k)+1} V^{(k+1)}_j$; $V(k+1)\leftarrow \bigcup^{k+1}_{l=0} \bigcup^{m(l)}_{j=m(l-1)+1} V^{(l)}_j$, go to 3.

3. If $k+1=t$, $T(t)\leftarrow T(k+1)$, go to 4; otherwise go to 2.

4. return $T(t)$ with a level function $l$.

\begin{figure}[h]
\centering
\includegraphics[height=6cm]{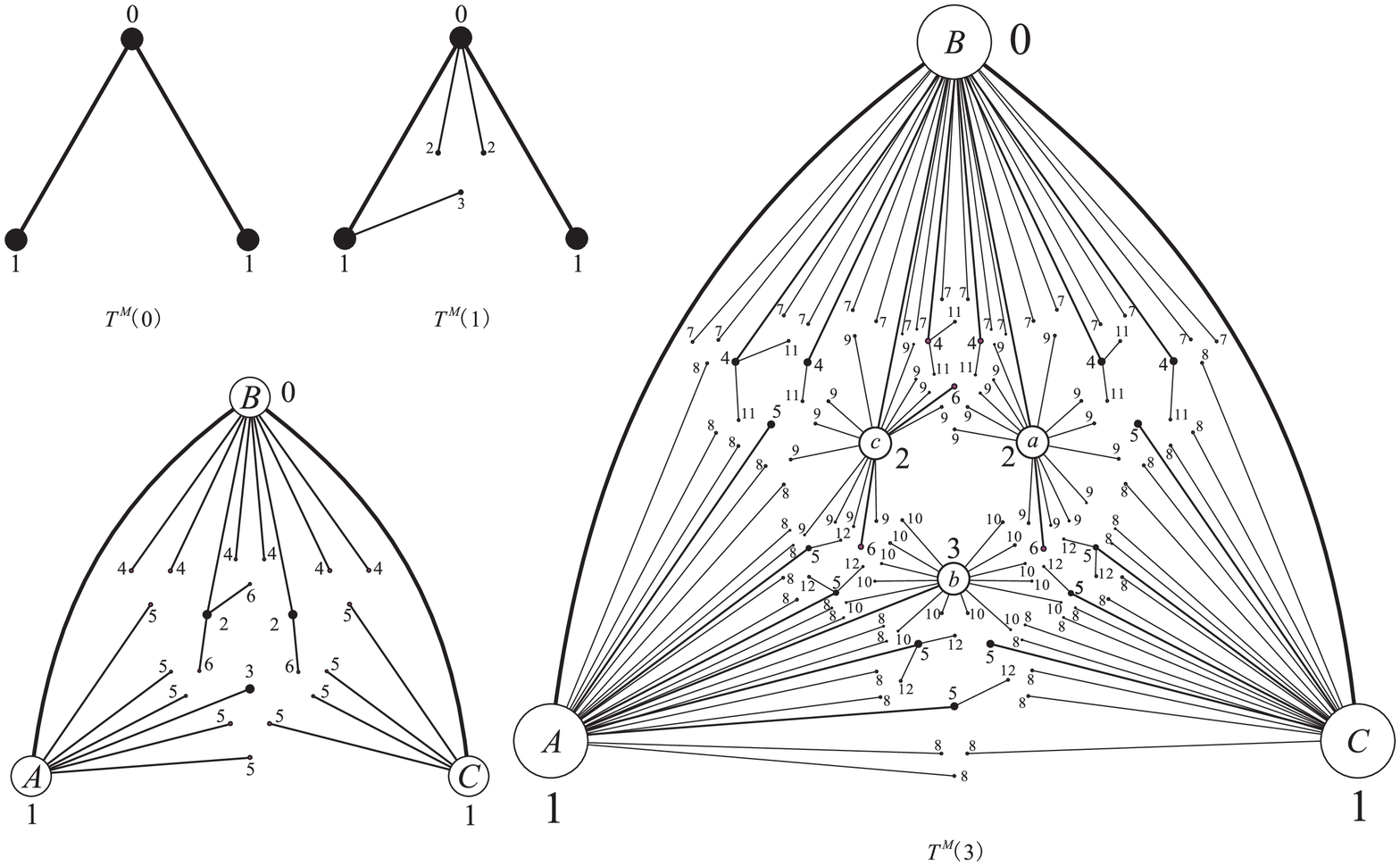}
\caption{The DFF-BFSA algorithm produces four MLS-trees with level functions in $N(k)$ for $k=0,1,2,3$.}
\label{fig.DFF-BFSA}
\end{figure}

\begin{thm} \label{them:build-maxi-leaf-spanning-trees}
The DFF-BFSA algorithm can find an MLS-tree $T$ of a GSN-model $N(t)$ such that $T$ has $19\cdot 6^{t-2}$ leaves, diameter $D(T)= 2t$ and maximum degree $\Delta(T)=\Delta(t)$ with $t\geq 2$.
\end{thm}
\begin{proof} Suppose that there is a cycle $C=x_1x_2\cdots x_mx_1$ in the graph $T$ obtained by the DFF-BFSA algorithm. Without loss of generality, the level values $l(x_1)<l(x_i)$ for $x_i\neq x_1$ in $C$. So there are three nodes $u,v,w$ such that $l(u)<l(w)$ and $l(v)<l(w)$, which mean that $w\in \textrm{nei}(u)\setminus \textrm{nei}(s)$ and $w\in \textrm{nei}(v)\setminus N(k)$, but it is impossible since $(\textrm{nei}(u)\setminus \textrm{nei}(s))\cap (N(v)\setminus N(k))=\emptyset$. Thereby, $T$ contains no cycle (Ref. $y\in \textrm{nei}(x^{(l)}_{j,i},k+1)\setminus \{w\in V(t): l(w)\textrm{ exists}\}$). Notice that $N(t)$ is connected, and at each time step the DFF-BFSA algorithm scans all neighbors of a node. So, $T$ is connected and a spanning tree.

We select the first node $u_0=B$ in $N(2)$, so $l(u_0=B)=0$ and degree $k(u_0)$ is equal to maximum degree  $\Delta(k)=3^{k}+1$  of $N(k)$ for $2\leq k\leq t$ by the DFF-BFSA algorithm. Notice that $l(A)=l(C)=1$, so $|k(A)-k(C)|=1$ and $\max\{k(A),k(C)\}=k(u_0)-1=3^{t}$. Three nodes $A,B,C$ of the spanning tree $T$ control other $3^{k+1}+1$ nodes  of $N(k)$. For $t=0,1,2$, it is not hard to see diameters $D(T)=2t$. For $t\geq 3$, by the DFF-BFSA algorithm, every path $P(A,w)$ from node $A$ to a leaf $w$ has at most length $(t-1)$ if it does not pass through node $B$, and each path $P(C,w')$ from node $C$ to a leaf $w'$ has at most length $(t-1)$ if it does not pass through node $B$. Thereby, the path from $w$ to $w'$ has length $2t$, which means $D(T)=2t$.

Notice that the spanning trees have $19\cdot 6^{t-2}$ leaves for $t=0,1,2,3$ (Ref. Figure \ref{fig.DFF-BFSA}). We can confirm that for $t\geq 4$, every spanning tree of $N(t)$ obtained by the DFF-BFSA algorithm has $19\cdot 6^{t-2}$ leaves according to Theorem \ref{them:basic-maximum-leaf-trees}. The theorem is covered.
\end{proof}

\section{Conclusion}

For determining the kernels of GSN-models we focus on MLS-trees of GSN-models, and show the structures of some MLS-trees by our algorithms. Clearly, all MLS-trees $T^M(i+1)$ having the shortest diameter $2(i+1)$ can be constructed by our generalized MLS-TREE algorithm over all MLS-trees having shortest diameter $D(T^M(i))=2i$ in $M_{st}(i)$. Although our DFF-BFSA algorithm can not find spanning trees having maximal leaves in any growing network model, however, we verify it for some growing network models, and find out some interesting spanning trees.

Suppose that a triangle $\Delta xbc$ and another triangle $\Delta bcy$ have a common edge $bc$ in a maximal planar graph $G$ whose faces are triangular. We remove the edge $bc$ and then join $x$ with $y$ by an edge, the resulting is still a maximal planar graph, written as $G'$ and say `flipping the edge $bc$'. We call the procedure of obtaining $G'$ from $G$ a \emph{flip operation}. Note that  every GSN-model $N(t)$ is a maximal planar graph, and ``Any pair of maximal planar graphs on $n$ vertices can be transformed into each other by at most $5.2n-24.4$ flip operations (Ref. \emph{\cite{Mori-Nakamoto-Ota2003}}).''  We propose a problem: \emph{For what value of a positive integer $m$, does rewiring $m$ edges of a GSN-model $N(t)$ by the flip operation produce a scale-free network model?} For larger integers $t>0$, we guess that a maximal planar graph $H^*$ obtained from a  GSN-model $N(t)$ by flipping some edges having ends in $V(t)\setminus V(t-3)$ is scale-free. The above problem leads to a problem of graph theory: \emph{Determine finite maximal planar graphs $G_1,G_2,\dots $ such that $G_{t-1}$ is a proper subgraph of $G_t$ and each $G_t$ obeys a power law distribution}.

\vskip 0.6cm

\textbf{Acknowledgments}\quad
This research was supported by the National Natural Science
Foundation of China under Grant Nos. 61163054, 61363060 and 61163037.

\end{document}